\documentclass[draft,12pt,reqno,a4paper]{article}

\usepackage[margin=3cm,footskip=1cm]{geometry}
\usepackage{amssymb} 
\usepackage{amsmath,amsthm} 
\usepackage{mathtools}
\usepackage{mathrsfs,bm}

\usepackage{enumerate} 
\usepackage{color}

\newcommand{\ri}{{\rm i}} 

\newcommand{\N}{{\mathbb{N}}} \newcommand{\R}{{\mathbb{R}}} 
\newcommand{\C}{{\mathbb{C}}} 

\newcommand{\vA}{{\mathcal A}} 
\newcommand{\vD}{{\mathcal D}}

\theoremstyle{plain}
\newtheorem{thm}{Theorem}[section]

\newtheorem{corollary}[thm]{Corollary}
\theoremstyle{definition} 
 
\newtheorem{cond}[thm]{Condition}
\newtheorem{defi}[thm]{Definition}

\newtheorem{remarks}[thm]{Remarks}
\newtheorem*{remarks*}{Remarks}
\newtheorem*{remark*}{Remark}

\numberwithin{equation}{section}

\title{Resonance free domain for Schr\"odinger operators with repulsive potential}

\author{Kyohei Itakura\thanks{
Graduate School of Mathematical Science, 
The University of Tokyo,\ Tokyo,\ Japan.\ 
E-mail: itakura@ms.u-tokyo.ac.jp}}

\date{}
\begin{document}

\allowdisplaybreaks
\maketitle

\begin{abstract}
We study resonances for the Schr\"odinger operators 
with quadratic or sub-quadratic repulsive potential. 
In the present paper, we show the non-existence of resonances 
in some complex neighborhood of a fixed energy 
by employing schemes of Briet--Combes--Duclos and Hunziker 
and by introducing a proper distortion based on, but slightly modified, 
positive commutator method of Mourre. 
\end{abstract}




\section{Introduction}\label{sec:Intro}

In the present paper we consider semiclassical repulsive Hamiltonians: For a fixed $s\in(0,1]$
\begin{align*}
H=H_0+q(x);\quad 
H_0=-\tfrac{\hbar^2}2\Delta-\tfrac12|x|^{2s}
\end{align*}
on $L^2(\R^d)$, 
where $\hbar>0$ is the semiclassical parameter, $x\in\R^d$ and $q$ is a perturbation. 
If $s=1$, $H_0$ is called the inverted harmonic oscillator. 
We define resonances of $H$ as eigenvalues of complex distorted $H$, 
which is denoted by $H_\theta$ later, 
and show that for sufficiently small $\hbar>0$, 
$H$ has no resonances in some complex neighborhood of a fixed energy $E$ 
under a kind of \emph{non-trapping condition} for potential $-\tfrac12|x|^{2s}+q$. 
For the case of perturbed Laplacian 
\begin{equation*}
P=-\tfrac12\Delta+V(x), 
\end{equation*}
where $V$ is a decaying real-valued smooth function, 
one assumes as a non-trapping condition that the quantity
\begin{equation*}
2(V-E)+x\cdot \partial V
\end{equation*}
is bounded from above by a negative constant outside of classically forbidden region. 
Such a condition is called virial condition. 
This quantity comes from the Poisson bracket
\begin{equation*}
\{x\cdot\xi, P^{\rm{cl}}\}|_{P^{\rm{cl}}=E},
\end{equation*}
where $P^{\rm{cl}}=\tfrac12\xi^2+V(x)$ is the classical Hamiltonian of $P$.
The choice of $x\cdot\xi$ is related to positive commutator method of \cite{Mo}. 
In fact, $\mathop{\mathrm{Re}}(x\cdot p)$ with $p=-\ri\nabla$ 
plays a role of the conjugate operator for $P$. 
However $\mathop{\mathrm{Re}}(x\cdot p)$ is not the conjugate operator 
for the repulsive Hamiltonian $H$.
By considering the classical orbit of the particle in the repulsive electric field, 
we can take the operator $\mathop{\mathrm{Re}}((\partial g)\cdot p)$ with 
\begin{equation*}
g=
\begin{cases}
\tfrac1{2(1-s)}|x|^{2-2s} & \text{for}\ s<1, \\
\tfrac12(\log |x|)^2 & \text{for}\ s=1, 
\end{cases}
\end{equation*}
as the conjugate operator for the repulsive Hamiltonian $H$. 
Therefore, it would be better to consider $(\partial g)\cdot\xi$ instead of $x\cdot\xi$, 
see also Remarks~\ref{thm:remarks} (2). 

In the proof of our main result, we employ schemes of \cite{BCD, H}, 
which are used for a perturbed Laplacian. 
As the author mentioned above, 
we slightly modify the quantity $x\cdot\xi$ appeared in their schemes. 
We discuss eigenvalue of distorted $H$, not dilated, which is 
constructed by using a complex distortion 
related to the flow of $(\partial g)\cdot\nabla$. 

There is a large body of literature on resonances 
for semiclassical Schr\"odinger operators with bounded potentials. 
We refer to \cite{AC, BCD, H, Kl, Ma} for non-existence of resonances in a certain region. 
We study resonance free domain for $H$ by employing their methods. 
We also refer to \cite{BFRZ, LBMa, N} for asymptotic behavior of resonances 
and to \cite{Z} in which many topics of resonances are dealt with. 
Resonances of the repulsive Hamiltonian $H$ can be treated in the framework of \cite{HeSj}. 
However, thanks to using our new distortion 
we can discuss resonances for $H$ with a large class of perturbations compared to \cite{HeSj}, 
cf. Remarks~\ref{thm:remarks} (3). 
The Stark potential is a kind of repulsive potential, although it is not spherically symmetric.  
Resonances in the Stark effect is studied for many authors, 
we refer to e.g. \cite{HiSi, Ka, Ko, W}.

\section{Setting and Results}\label{sec:SandR}
\subsection{Setting}

First we assume a \emph{long-range type} condition on $q$.

\begin{cond}\label{cond:long-range-type}
The perturbation $q$ belongs to $C^\infty(\R^d; \R)$.
Moreover there exist $\rho\in(0,1)$ and $C_k>0$ for $k\in\N\cup\{0\}$ such that 
\begin{equation*}
|\nabla^k q| \le 
\begin{cases}
C_k \langle x\rangle^{2s-k-\rho} & \text{for}\ s<1, \\
C_k (\log\langle x\rangle)^{-1-\rho}\langle x\rangle^{2-k} & \text{for}\ s=1.
\end{cases}
\end{equation*}
\end{cond}

Although $H$ has a self-adjoint extension under Condition~\ref{cond:long-range-type} 
by the Faris--Lavine theorem cf. \cite{RS}, 
let us consider $H$ as the self-adjoint operator associated with the quadratic form on 
$H^1\cap\vD(|x|^s)$ 
\begin{equation*}
t[u]:=\tfrac{\hbar^2}2\|\nabla u\|^2-\tfrac12\||x|^su\|^2+\langle u, q(x)u\rangle.
\end{equation*}

In the study of resonances for perturbed Laplacian $P$, 
It would be standard to use the dilation $x\mapsto \mathrm{e}^\theta x=(\mathrm{e}^\theta r)\omega$ 
or slightly modified one. 
Here $r=|x|$ and $\omega=x/r$. 
We note that the dilation generates an operator $\mathop{\mathrm{Re}}(x\cdot p)$ 
which is the conjugate operator for perturbed Laplacian $P$ in the Mourre theory. 
As for the case of sub-quadratic repulsive Hamiltonian, that is $s\in(0,1)$, 
an operator $\mathop{\mathrm{Re}}(|x|^{-2s}x\cdot p)$ plays a role of the conjugate operator, 
and this is generated by a group $T(\theta): x\mapsto (r^{2s}+2s\theta)^{1/(2s)}\omega$ 
for $\theta\ge0$.
Therefore, in order to consider resonances for sub-quadratic repulsive Hamiltonian, 
it is natural to use such a transformation. 
As for the case of the inverted harmonic oscillator, that is $s=1$, 
an operator $\mathop{\mathrm{Re}}((\log|x|)|x|^{-2}x\cdot p)$ is chosen as the conjugate operator. 
In order to find a group which generates this operator we need to solve the equation 
\begin{equation*}
\tfrac{\partial r(\theta)}{\partial \theta}=(\log r)r^{-1}\tfrac{\partial r(\theta)}{\partial r}.
\end{equation*}
However it is difficult to find an exact solution to this equation. 
Here we note that we can discuss resonances of the inverted harmonic oscillator 
by using an \emph{approximate} solution $r(\theta)=(r^2+2\theta\log r)^{1/2}$ 
in the sense that this $r(\theta)$ satisfies
\begin{equation*}
\tfrac{\partial r(\theta)}{\partial \theta}-(\log r)r^{-1}\tfrac{\partial r(\theta)}{\partial r}
=-\theta r^{-2}\log r.
\end{equation*}
However we need an additional assumption on the classically forbidden region to control an error term.

Based on the above short discussion, 
let us define a complex distortion by mimicking an argument of \cite{H}. 
We take a smooth cut-off function $\chi_1: \R\to[0,1]$ which satisfy 
\begin{equation*}
\chi_1(t)=\begin{cases}
0 & \ \text{for}\ t\le 1, \\
1 & \ \text{for}\ t\ge 2.
\end{cases}
\end{equation*}
We note that $\chi_1$ is a Lipschitz continuous function, 
and we denote its Lipschitz constant by $L>0$:
\begin{equation*}
|\chi_1(t_1)-\chi_1(t_2)|\le L|t_1-t_2|.
\end{equation*}
We set for $R>0$ 
\begin{equation*}
\chi_R(t):=\chi_1(t/R^{2s}).
\end{equation*}
For $\theta\in\R$ we introduce the mappings
\begin{equation}\label{eq:chiR}
r\mapsto r_\theta:=
\begin{cases}
\left(r^{2s}+2s\theta\chi_R(r^{2s})\right)^{1/(2s)} & \text{for}\ s<1, \\
\left(r^{2}+2\theta\chi_R(r^{2})\log r\right)^{1/2} & \text{for}\ s=1.
\end{cases}
\end{equation}
By the Lipschitz continuity of $\chi_R$, 
the mapping $r\mapsto\left(r^{2s}+2s\theta\chi_R(r^{2s})\right)^{1/(2s)}$ 
is invertible for any $\theta\in\R$ with $|\theta|<(2s)^{-1}R^{2s}L^{-1}$.
In fact, for any $\tilde r\in[0,\infty)$ we can construct $r\ge0$ such that 
\begin{equation}\label{eq:2301152225}
\tilde{r}^{2s}=r^{2s}+2s\theta\chi_R(r^{2s}).
\end{equation}
For any $\tilde r\in[0,\infty)$, let us consider the sequence $\{r_n\}_{n\in\N_0}\subset[0,\infty)$ 
defined by the recurrence formula
\begin{equation*}
r_n^{2s}=\tilde{r}^{2s}-2s\theta\chi_R(r_{n-1}^{2s}),\quad r_0=\tilde r.
\end{equation*}
Noting that $\sup_{n\in\N_0}r_n^{2s}\le \tilde{r}^{2s}+2s|\theta|$ 
and for any $n,m\in\N_0, n\ge m$
\begin{align*}
|r_n^{2s}-r_m^{2s}| 
&= 
2s|\theta||\chi_R(r_{n-1}^{2s})-\chi_R(r_{m-1}^{2s})|
\\&\le 
2s|\theta|R^{-2s}L|r_{n-1}^{2s}-r_{m-1}^{2s}|
\\&\phantom{{}ii{}}\vdots 
\\&\le 
(2s|\theta|R^{-2s}L)^m|r_{n-m}^{2s}-r_0^{2s}|,
\end{align*}
we can see that $\{r_n\}_{n\in\N_0}$ is a Cauchy sequence, and thus it converges. 
By letting $r=\lim_{n\to\infty}r_n$, we have \eqref{eq:2301152225}.
Similarly, we can verify that there exists $L_1>0$ such that the lower mapping of \eqref{eq:chiR} 
is invertible for any $\theta\in\R$ with $|\theta|<L_1$. 
For simplicity, we let $L_s=(2s)^{-1}R^{2s}L$ for $s\in(0,1)$. 

For $\theta\in\R$ with $|\theta|<L_s$, 
we define unitary operator $U_\theta$ on $L^2(\R^d)$ as
\begin{equation}\label{eq:Utheta}
(U_\theta u)(x)=(U_\theta u)(r\omega):=J^{1/2}u(r_\theta\omega), 
\end{equation}
where 
\begin{equation*}
J=J(\theta, r)=
\begin{cases}
(r_\theta/r)^{d-2s}(1+2s\theta\chi_R'(r^{2s})) & \text{for}\ s<1, \\
(r_\theta/r)^{d-2}(1+\theta r^{-2}\chi_R(r^2)+2\theta\log r\chi_R'(r^2)) & \text{for}\ s=1.
\end{cases}
\end{equation*}

Now we introduce an additional condition on $q$.
For any $\theta\in\R$, we set 
\begin{equation*}
q_\theta(x)=q_\theta(r\omega)=q(r_\theta\omega). 
\end{equation*}
\begin{cond}\label{cond:analytic}
In addition to Condition~\ref{cond:long-range-type}, 
for a fixed $R>0$, $q_\theta$ has an analytic extension in a region 
$C_{\beta_0}:=\{\theta\in\C \,|\, |\mathop{\mathrm{Im}}\theta|<\beta_0\}$ 
for some $\beta_0>0$ uniformly in $x\in\R^d$.
\end{cond}
In the statement of our main theorem, 
we fix $R>0$ depending on the situation and then assume Condition~\ref{cond:analytic}. 
We remark that, as we will see later, the resonance of $H$ does not depend on the choice of $R>0$.

We set for $\theta\in\R$ with $|\theta|<L_s$
\begin{align*}
t_\theta[u] 
&:=
t[U_\theta^{-1}u] 
= 
\tfrac{\hbar^2}2\langle u, U_\theta p^2U_\theta^{-1}u\rangle 
-\tfrac12\langle u, r_\theta^{2s}u\rangle 
+\langle u, q_\theta u\rangle. 
\end{align*}
We note that $U_\theta (-\Delta)U_\theta^{-1}$ is expressed as 
\begin{align*}
U_\theta (-\Delta)U_\theta^{-1}
&=
U_\theta \bigl(-\partial_r^2-\tfrac{d-1}r\partial_r-r^{-2}\Delta_{\mathbb{S}^{d-1}}\bigr)U_\theta^{-1}
\\&= 
-J^{1/2}(\tfrac{\partial r_\theta}{\partial r})^{-1}\partial_r
(\tfrac{\partial r_\theta}{\partial r})^{-1}\partial_rJ^{-1/2}
-(d-1)J^{1/2}r_\theta^{-1}(\tfrac{\partial r_\theta}{\partial r})^{-1}\partial_rJ^{-1/2}
\\&\phantom{{}={}}
-r_\theta^{-2}\Delta_{\mathbb{S}^{d-1}}.
\end{align*}
Then we can see that $t_\theta[u]$ extends analytically to a complex region 
$\{\theta\in\C \,|\, |\theta|<L_s\}\cap C_{\beta_0}$ 
and forms analytic family of operators 
\begin{equation*}
H_\theta:=U_\theta HU_\theta^{-1}
=\tfrac{\hbar^2}2U_\theta (-\Delta)U_\theta^{-1}-\tfrac12r_\theta^{2s}+q_\theta.
\end{equation*}

Based on \cite[Section~6]{H}, we consider a set of \emph{analytic vectors} 
to characterize resonances of $H$. 
Let us introduce the set $F$ of entire function $f=f(z)$, $z\in\C^d$, 
which vanishes faster than any inverse power of $|\mathop{\mathrm{Re}}z|$ as 
$|\mathop{\mathrm{Re}}z|\to\infty$ in any region 
\begin{equation*}
|\mathop{\mathrm{Im}}z|\le (1-\varepsilon)|\mathop{\mathrm{Re}}z|,\quad \varepsilon\in(0,1).
\end{equation*}
We call $\psi\in L^2(\R^d)$ analytic vector if it holds that $\psi(x)=f(x)$ on $\R^d$ for some $f\in F$. 
We denote the set of analytic vectors by $\vA$.
Then we can see, similarly to \cite{H}, 
that there exist $\tilde L_s\in(0,L_s]$ such that for any $\psi\in\vA$ 
the mapping $\theta\to U_\theta \psi$ is an $L^2$-valued analytic function in $|\theta|<\tilde L_s$.
Moreover for any $\theta\in\C$ with $|\theta|<\tilde L_s$, 
$\vA$ and $U_\theta\vA$ are dense in $L^2(\R^d)$. 
For any $z\in\C$ with $\mathop{\mathrm{Im}}z>0$ and any real $\theta$ with $|\theta|<L_s$ 
we have  
\begin{equation}\label{eq:2310191507}
\langle \phi, (H-z)^{-1}\psi\rangle 
=\langle U_{\bar\theta}\phi, (H_\theta-z)^{-1}U_\theta\psi\rangle, \quad \phi,\psi\in\vA.
\end{equation}
In particular it has a meromorphic continuation in $z$ to $\C\setminus\sigma_{\rm ess}(H_\theta)$ 
and in $\theta$ to $\{\theta\in\C\,|\,|\theta|<\tilde L_s\}$. 
In general, resonance is defined as a pole of a meromorphic continuation of \eqref{eq:2310191507}. 
In the present paper, we adopt the following definition of resonances, which is equivalent to it.
\begin{defi}\label{thm:def-resonance}
We call $z\in\C$ a resonance of $H$ 
if $z$ is an eigenvalue of $H_\theta$ for some $\theta\in\ri\R_+$.
\end{defi}
We note that the resonances does not depend on the choice of the cut-off function $\chi_R$ 
by uniqueness of meromorphic continuation of \eqref{eq:2310191507}. 
Moreover it does not depend on the value of $\theta=\ri\beta$, $\beta>0$, 
in the sense that if $z$ is an eigenvalue of $H_{\ri\beta_1}$, 
$z$ is also an eigenvalue of $H_{\ri\beta_2}$ when $\beta_2>\beta_1$.

\subsection{Resonance free domain}

The following theorem is our main result 
and it asserts that $H$ has no resonances 
in some complex neighborhood of fixed energy $E$ for sufficiently small $\hbar>0$. 
\begin{thm}\label{thm:no-resonance}
Let $s\in(0,1]$ and take any $E\in\R$. 
Assume that there exist $\mu>0$ 
and smooth cut--off functions $\chi, \tilde{\chi}\in C^1([0,\infty); [0,1])$ such that 
\begin{enumerate}[(i)]
\item $\chi^2+\tilde\chi^2=1 \ \text{on}\ \R^d$ and $\chi=1$ near $0$. 

Moreover, if $s=1$, $\mathop{\mathrm{supp}}\tilde\chi\subset\{r\ge\sqrt{\mathrm{e}}\}$,
\item There exists $\alpha>0$ such that 
$$
-\tfrac12r^{2s}+q-E \ge \alpha\ \ 
\text{on}\ \mathop{\mathrm{supp}}\chi,
$$
\item There exists $\gamma>0$ such that 
$$
\begin{cases}
1-s-2(1-2s)r^{-2s}(q-E)-r^{1-2s}(\partial_rq)-\mu \ge \gamma & \text{for}\ s<1, \\ 
1+2r^{-2}(\log r-1)(q-E)-r^{-1}(\log r)(\partial_rq)-\mu \ge \gamma & \text{for}\ s=1 
\end{cases}
$$
on $\mathop{\mathrm{supp}}\tilde\chi$.
\end{enumerate} 
Take any $R>0$ such that $\chi_R=1$ on $\mathop{\mathrm{supp}}\tilde\chi$, 
and assume Condition~\ref{cond:analytic}.
Let $\beta\in(0, \min\{\tilde L_s,\beta_0\})$, $\theta=\ri\beta$ and $z=E-\ri\beta\mu$. 
Then one has for any $u\in \vD(H_\theta)$ and for any $\beta>0$ small enough
\begin{equation*}
\|(H_\theta-z)u\|\ge C_{\alpha, \beta, \gamma, \hbar}\|u\|
\end{equation*}
with 
$C_{\alpha, \beta, \gamma, \hbar}=
(1-c\beta)\min\{\alpha, \beta\gamma\}+O(\beta^2)+O(\hbar^2)$.
Here $c>0$ is a certain constant which is independent of both $\hbar$ and $\beta$.
\end{thm}

\begin{corollary}\label{cor:no-resonance}
Suppose all the assumptions of Theorem~\ref{thm:no-resonance}. 
Then there exists $\beta_c$ such that for any each $\beta\in(0,\beta_c)$ 
there exists $\hbar_0>0$ such that for any $\hbar\in(0, \hbar_0)$ 
the constant $C_{\alpha, \beta, \gamma, \hbar}$ is strictly positive, 
and hence, $z=E-\ri\mu\beta$ is not a resonance 
of $H=-\tfrac{\hbar^2}2\Delta-\tfrac12|x|^{2s}+q(x)$.
\end{corollary}

\begin{remarks}\label{thm:remarks}
(1) Let us take the sequence $\{\phi_{s,\lambda,n}\}_{n\in\N}\subset L^2(\R^d)$: 
\begin{equation*}
\phi_{s,\lambda,n}=
c_{s,n}\eta_n(f) r^{-(d+s-1)/2}\exp\{\ri(\tfrac1{1+s}r^{1+s}+\lambda f+\theta(r))/\hbar\},
\end{equation*}
where $\lambda\in\R$, $c_{s,n}$ is a normalizing constant, $\eta_n$ is a smooth cut-off function 
which is obeying $\mathop{\mathrm{supp}}\eta_n\subset[2^n, 2^{n+1}]$ and 
\begin{equation*}
f=f(r)=\begin{cases}
\tfrac1{1-s}(r^{1-s}-1)+1 & \text{for}\ s<1, \\
\log r +1 & \text{for}\ s=1.
\end{cases}
\end{equation*}
Moreover $\theta(r)$ is a certain smooth function which depends on $s$ and $\rho$, and satisfy 
\begin{equation*}
(r^s+\lambda r^{-s}+\partial_r\theta(r))^2-(r^{2s}+\lambda-q)=o(1)
\end{equation*} 
as $r\to\infty$. 
For example, if $\rho>s$ it is sufficient to take $\theta(r)$ as 
\begin{equation*}
\theta(r)=-\int_2^r q(\tilde r\omega)\tilde r^{-s}\,{\rm d}\tilde r.
\end{equation*}
Then we can see that for $\theta=\ri\beta$ and for any $\lambda\in\R$
\begin{align*}
\|(H_\theta-(\lambda-(1-s)\beta\ri))\phi_{s,\lambda,n}\| &\to 0\quad \text{for}\ s<1,
\\
\|(H_\theta-(\lambda-\beta\ri))\phi_{s,\lambda,n}\| &\to 0\quad \text{for}\ s=1
\end{align*}
as $n\to\infty$, and hence 
the essential spectrum of $H_\theta$ with $\theta=\ri\beta$ is 
\begin{equation*}
\sigma_{\mathrm{ess}}(H_\theta)=
\begin{cases}
\{z\in\C\,|\,\mathop{\mathrm{Im}}z=-(1-s)\beta\} & \text{for}\ s<1, \\
\{z\in\C\,|\,\mathop{\mathrm{Im}}z=-\beta\} & \text{for}\ s=1. 
\end{cases}
\end{equation*}
The eigenvalues of $H_\theta$ would appear, if it exist, in the strip 
$\{z\in\C\,|\,0<-\mathop{\mathrm{Im}}z<(1-s)\beta\}$ for $s<1$, 
or $\{z\in\C\,|\,0<-\mathop{\mathrm{Im}}z<\beta\}$ for $s=1$. 

(2) The bound of (iii) is a virial type condition, and it comes from an estimate of the Poisson bracket 
\begin{align*}
\{h, a\}|_{h=E}
=\Bigl[\tfrac{\partial h}{\partial\xi}\cdot\tfrac{\partial a}{\partial x}
-\tfrac{\partial h}{\partial x}\cdot\tfrac{\partial a}{\partial\xi}\Bigr]\!\big|_{h=E},
\end{align*}
where $h=\tfrac12\xi^2-\tfrac12|x|^{2s}+q(x)$ is the classical mechanics version of $H$ and 
\begin{align*}
a=(\partial g)\cdot\xi\ \ \text{with}\ g=
\begin{cases}
\tfrac1{2(1-s)}r^{2-2s} & \text{for}\ s<1, \\
\tfrac12(\log r)^2 & \text{for}\ s=1. 
\end{cases}
\end{align*}
Since $g$ is non-negative and $g=O(t^2)$ as $t\to\infty$ in the classical mechanics sense, 
the quantity $\{h, a\}|_{h=E}$ is expected to be positive. 

(3) Our $r_\theta$ has satisfies
\begin{equation*}
r_\theta=(r^{2s}+2s\theta\chi_R(r^{2s}))^{2s}\to \mathrm{e}^\theta r\quad 
\text{as}\ \ s\downarrow0
\end{equation*}
for any $r>2R$. 
In this sense, we can regard our main result as a generalization of that for perturbed Laplacian.  
We can also discuss resonance free domain for the repulsive Hamiltonians $H$ by using the dilation. 
If we use the dilation we need to assume that $q_\theta$ is analytic in a sector 
$\{|\mathop{\mathrm{Im}}\mathrm{e}^\theta x|<c_1|\mathop{\mathrm{Re}}\mathrm{e}^\theta x|\}
\setminus \{|\mathop{\mathrm{Re}}\mathrm{e}^\theta x|<c_2\}$
for some $c_1, c_2>0$. 
However this sector is wider than the region 
in which analyticity of $q_\theta$ is required in Condition~\ref{cond:analytic}.
In fact, since we have for large $r=|x|$ 
\begin{equation*}
r_\theta\sim r(1+\theta\chi_R(r^{2s})r^{-2s}), 
\end{equation*}
$\mathop{\mathrm{Im}}r_\theta$ diverges slower than  
$\mathop{\mathrm{Im}}\mathrm{e}^\theta r$ as $r\to\infty$. 
Especially, if $s>1/2$, $\mathop{\mathrm{Im}}r_\theta$ decays as $r\to\infty$. 
Hence by using our distortion instead of the dilation we can deal with a large class of potentials. 
This is a novelty of the present paper.

We also mention the case of Stark Hamiltonian. 
Roughly speaking, Stark Hamiltonian is corresponding to the case of $s=1/2$. 
In \cite[Section~23.1]{HiSi} and \cite{Ka}, they study resonances by introducing a distortion 
$x\mapsto x+\theta v(x)$ for $\theta\in\C$ with $|\theta|\ll1$, 
where $v$ is a certain bounded smooth function. 
The size of both of the regions $\{x+\theta v(x)\,|\,x\in\R^d\}$ 
and $\{r_\theta(x)\,|\,x\in\R^d, s=1/2\}$ in $\C$ are almost the same.
\end{remarks}

\section{Examples}\label{sec:example}

This is a short section, which 
we discuss the choice of $\alpha, \gamma>0$ and $\chi, \tilde\chi\in C^1$ 
of the assumptions of the theorem 
for \emph{free} semiclassical repulsive Hamiltonian $H=H_0$.

First we let $s\in(0,1)$. 
We take any $E<0$.
If $s\ge 1/2$, the bound 
\begin{equation*}
1-s+2(1-2s)r^{-2s}E-\mu\ge \gamma
\end{equation*}
always hold for any $\mu\in(0,1-s)$ and $\gamma\in(0, 1-s-\mu)$. 
Then for any $\alpha\in(0, -E)$, 
we can easily construct functions $\chi$ and $\tilde\chi$ 
which obeying the assumptions of the theorem. 
Let us consider the case of $s<1/2$. 
For any $\mu\in(0,s)$ and $\alpha\in(0,-(1-\tfrac{1-2s}{1-s-\mu})E)$, 
we take and fix a constant $c\in(\tfrac{1-2s}{1-s-\mu}\tfrac{E}{E+\alpha},1)$. 
Then by taking smooth cut--off functions $\chi$ and $\tilde\chi$ which satisfy
\begin{align*}
\chi=\chi(r)=\begin{cases}
1 & \text{on}\ \{r=|x|\in\R\,|\,-\tfrac1{2c} r^{2s}-E>\alpha \}, \\
0 & \text{on}\ \{r=|x|\in\R\,|\,-\tfrac12r^{2s}-E\le\alpha \}, 
\end{cases}
\quad 
\tilde\chi=\sqrt{1-\chi^2},
\end{align*}
we obtain on $\mathop{\mathrm{supp}}\tilde\chi$ that for some $\gamma>0$ 
\begin{align*}
1-s+2(1-2s)r^{-2s}E-\mu
\ge 
1-s-(1-2s)\tfrac{E}{c(E+\alpha)}-\mu 
>\gamma. 
\end{align*}

Next we let $s=1$. 
In this case, we need restriction $E<-1/2$ to obtain the bound 
\begin{equation}\label{eq:2102061442}
1-2r^{-2}(\log r-1)E-\mu \ge \gamma\ \ \text{on}\ \mathop{\mathrm{supp}}\tilde\chi
\end{equation}
for some $\mu, \gamma>0$. 
For $\alpha\in(0, -(2E+1)/4)$ we introduce smooth functions $\chi$ and $\tilde\chi$ as 
\begin{align*}
\chi=\chi(r)=\begin{cases}
1 & \text{on}\ \{r=|x|\in\R\,|\,-\tfrac12 r^2-E>2\alpha \}, \\
0 & \text{on}\ \{r=|x|\in\R\,|\,-\tfrac12r^2-E\le\alpha \}, 
\end{cases}
\quad 
\tilde\chi=\sqrt{1-\chi^2}.
\end{align*}
When $E<-{\mathrm e}^2/2$, 
by retaking $\alpha\in(0, -(2E+{\mathrm e}^2)/4)$, 
we can see that the term $-2r^{-2}(\log r-1)E$ is non-negative on $\mathop{\mathrm{supp}}\tilde\chi$. 
Thus for any $\mu\in(0, 1)$ and $\gamma\in(0, 1-\mu)$, the bound \eqref{eq:2102061442} holds. 
Let $E\in[-{\mathrm e}^2/2, -1/2)$. 
Then for any $\mu\in(0, \tfrac12\log(-2E))$ there exist $\alpha>0$ and $\gamma>0$ such that 
the bound \eqref{eq:2102061442} holds. 
In fact, for sufficiently small $\alpha>0$ we have on $\mathop{\mathrm{supp}}\tilde\chi$ 
\begin{align*}
&1-2r^{-2}(\log r-1)E-\mu 
\\&\ge 
\tfrac12\log(-2E)-\mu
+\tfrac{2\alpha}{E+2\alpha}
+\tfrac{E}{2(E+2\alpha)}\log\left(\tfrac{E+2\alpha}E\right)
-\tfrac{\alpha}{E+2\alpha}\log(-2E)
\\&>0.
\end{align*}


\section{Proof of the theorem}\label{sec:proof}


\subsection{Sub-quadratic repulsive Hamiltonian}\label{sub-quad}

In this section we give a proof of Theorem~\ref{thm:no-resonance} for $s\in(0,1)$. 
\begin{proof}[Proof of Theorem~\ref{thm:no-resonance} with $s\in(0,1)$]
Let $\theta=\ri\beta$. 
Take any $u\in\vD(H_\theta)$, $\|u\|=1$. 
Then we set $v=\overline{(\tfrac{\partial r_\theta}{\partial r})}^2(\chi^2-\ri\tilde\chi^2)u$. 
Note that $\|v\|\le (1+O(\beta))\|u\|$.
We show that for $z=E-\ri\beta\mu$
\begin{equation}
\begin{split}\label{eq:lower-bound}
&\mathop{\mathrm{Re}}\langle v, (H_\theta-z)u\rangle 
\\&\ge 
(\min\{\alpha, \beta\gamma\}+O(\hbar^2)+O(\beta^2))\|u\|^2
+\beta^2\mathop{\mathrm{Re}}\langle u, \psi(r)(H_\theta-z)u\rangle, 
\end{split}
\end{equation}
where $\psi(r)$ is a certain bounded function. 
Then by combining with the inequalities
\begin{align*}
(1+O(\beta))\|u\|\|(H_\theta-z)u\| 
&\ge 
\|v\|\|(H_\theta-z)u\| 
\ge 
\mathop{\mathrm{Re}}\langle v, (H_\theta-z)u\rangle, 
\\ 
O(\beta^2)\|u\|\|(H_\theta-z)u\| 
&\ge 
-\beta^2\mathop{\mathrm{Re}}\langle u, \psi(r)(H_\theta-z)u\rangle, 
\end{align*}
we have the assertion.

The left-hand side of \eqref{eq:lower-bound} is expressed as 
\begin{equation}
\begin{split}\label{eq:2210111611}
\mathop{\mathrm{Re}}\langle v, (H_\theta-z)u\rangle
&=
\mathop{\mathrm{Re}}\langle (\chi^2-\ri\tilde\chi^2)u, 
(\tfrac{\partial r_\theta}{\partial r})^2(H_\theta-z)u\rangle
\\&=
\tfrac{\hbar^2}2\mathop{\mathrm{Re}}\langle \chi^2u, 
(\tfrac{\partial r_\theta}{\partial r})^2U_\theta (-\Delta)U_\theta^{-1}u\rangle
\\&\phantom{{}={}}
-\tfrac{\hbar^2}2\mathop{\mathrm{Im}}\langle \tilde\chi^2u, 
(\tfrac{\partial r_\theta}{\partial r})^2U_\theta (-\Delta)U_\theta^{-1}u\rangle
\\&\phantom{{}={}}
+\mathop{\mathrm{Re}}\langle \chi u, 
(\tfrac{\partial r_\theta}{\partial r})^2(-\tfrac12r_\theta^{2s}+q_\theta-z)\chi u\rangle
\\&\phantom{{}={}}
-\mathop{\mathrm{Im}}\langle \tilde\chi u, 
(\tfrac{\partial r_\theta}{\partial r})^2(-\tfrac12r_\theta^{2s}+q_\theta-z)\tilde\chi u\rangle.
\end{split}
\end{equation}
Let us further compute each term on the right-hand side of \eqref{eq:2210111611}. 
First we compute the factor 
$(\tfrac{\partial r_\theta}{\partial r})^2U_\theta (-\Delta)U_\theta^{-1}$. 
\begin{align*}
&(\tfrac{\partial r_\theta}{\partial r})^2U_\theta (-\Delta)U_\theta^{-1}
\\&= 
-(\tfrac{\partial r_\theta}{\partial r})^2J^{1/2}(\tfrac{\partial r_\theta}{\partial r})^{-1}\partial_r
(\tfrac{\partial r_\theta}{\partial r})^{-1}\partial_rJ^{-1/2}
\\&\phantom{{}={}}
-(d-1)(\tfrac{\partial r_\theta}{\partial r})^2
J^{1/2}r_\theta^{-1}(\tfrac{\partial r_\theta}{\partial r})^{-1}\partial_rJ^{-1/2}
-(\tfrac{\partial r_\theta}{\partial r})^2r_\theta^{-2}\Delta_{\mathbb{S}^{d-1}}
\\&= 
-J^{1/2}(\tfrac{\partial r_\theta}{\partial r})\partial_r
(\tfrac{\partial r_\theta}{\partial r})^{-1}J^{-1/2}\partial_r
-J^{1/2}(\tfrac{\partial r_\theta}{\partial r})\partial_r
(\tfrac{\partial r_\theta}{\partial r})^{-1}(\partial_r J^{-1/2})
\\&\phantom{{}={}}
-(d-1)(\tfrac{\partial r_\theta}{\partial r})(r_\theta/r)^{-1}r^{-1}\partial_r
+\tfrac12(d-1)J^{-1}(\tfrac{\partial r_\theta}{\partial r})(r_\theta/r)^{-1}r^{-1}(\partial_rJ)
\\&\phantom{{}={}}
-(\tfrac{\partial r_\theta}{\partial r})^2(r_\theta/r)^{-2}r^{-2}\Delta_{\mathbb{S}^{d-1}}
\\&= 
-\partial_r^2
-J^{1/2}(\tfrac{\partial r_\theta}{\partial r})
\bigl(\partial_r(\tfrac{\partial r_\theta}{\partial r})^{-1}J^{-1/2}\bigr)\partial_r
\\&\phantom{{}={}}
+\tfrac12J^{-1}(\partial_r J)\partial_r
+\tfrac12J^{1/2}(\tfrac{\partial r_\theta}{\partial r})
\bigl(\partial_r(\tfrac{\partial r_\theta}{\partial r})^{-1}(\partial_rJ)J^{-3/2}\bigr)
\\&\phantom{{}={}}
-(d-1)(\tfrac{\partial r_\theta}{\partial r})(r_\theta/r)^{-1}r^{-1}\partial_r
+\tfrac12(d-1)J^{-1}(\tfrac{\partial r_\theta}{\partial r})(r_\theta/r)^{-1}r^{-1}(\partial_rJ)
\\&\phantom{{}={}}
-(\tfrac{\partial r_\theta}{\partial r})^2(r_\theta/r)^{-2}r^{-2}\Delta_{\mathbb{S}^{d-1}}
\\&= 
-\Delta 
+J^{-1}(\partial_rJ)\partial_r
+(\tfrac{\partial r_\theta}{\partial r})^{-1}(\tfrac{\partial^2 r_\theta}{\partial r^2})\partial_r
-\tfrac34J^{-2}(\partial_rJ)^2
+\tfrac12J^{-1}(\partial_r^2J)
\\&\phantom{{}={}}
-\tfrac12J^{-1}(\tfrac{\partial r_\theta}{\partial r})^{-1}(\tfrac{\partial^2 r_\theta}{\partial r^2})
(\partial_r J)
-(d-1)\!\left((\tfrac{\partial r_\theta}{\partial r})(r_\theta/r)^{-1}-1\right)\!r^{-1}\partial_r
\\&\phantom{{}={}}
+\tfrac12(d-1)J^{-1}(\tfrac{\partial r_\theta}{\partial r})(r_\theta/r)^{-1}r^{-1}(\partial_rJ)
-\left((\tfrac{\partial r_\theta}{\partial r})^2(r_\theta/r)^{-2}-1\right)\!
r^{-2}\Delta_{\mathbb{S}^{d-1}}.
\end{align*}
For notational simplicity, we introduce 
\begin{equation*}
\begin{split}
\phi(r) 
&= 
-\tfrac34J^{-2}(\partial_rJ)^2
+\tfrac12J^{-1}(\partial_r^2J)
-\tfrac12J^{-1}(\tfrac{\partial r_\theta}{\partial r})^{-1}(\tfrac{\partial^2 r_\theta}{\partial r^2})
(\partial_r J)
\\&\phantom{{}={}}
+\tfrac12(d-1)J^{-1}(\tfrac{\partial r_\theta}{\partial r})(r_\theta/r)^{-1}r^{-1}(\partial_rJ),
\end{split}
\end{equation*}
and then, we can write
\begin{equation}
\begin{split}\label{eq:2301201357}
(\tfrac{\partial r_\theta}{\partial r})^2U_\theta (-\Delta)U_\theta^{-1}
&= 
-\Delta 
+J^{-1}(\partial_rJ)\partial_r
+(\tfrac{\partial r_\theta}{\partial r})^{-1}(\tfrac{\partial^2 r_\theta}{\partial r^2})\partial_r
\\&\phantom{{}={}}
-(d-1)\!\left((\tfrac{\partial r_\theta}{\partial r})(r_\theta/r)^{-1}-1\right)\!r^{-1}\partial_r
+\phi(r) 
\\&\phantom{{}={}}
-\left((\tfrac{\partial r_\theta}{\partial r})^2(r_\theta/r)^{-2}-1\right)\!
r^{-2}\Delta_{\mathbb{S}^{d-1}}.
\end{split}
\end{equation}
We remark that $\phi$ is a complex-valued smooth function having bounded derivatives.
By the identity
\begin{equation*}
\mathop{\mathrm{Re}}(\chi^2\Delta)
=\chi \Delta\chi+|\nabla\chi|^2,
\end{equation*} 
we have 
\begin{equation}\label{eq:2210121201}
-\tfrac{\hbar^2}2\mathop{\mathrm{Re}}\langle \chi^2u, \Delta u\rangle
=\tfrac{\hbar^2}2\left(\|\nabla(\chi u)\|^2-\|(\nabla\chi)u\|^2\right).
\end{equation}
Noting that 
\begin{align*}
\partial_rJ 
&= 
\!\left(\partial_r(r_\theta/r)^{d-2s}(1+2s\theta\chi_R'(r^{2s}))\right)
\\&= 
(d-2s)(r_\theta/r)^{d-2s-1}(\partial_r\tfrac{r_\theta}{r})(1+2s\theta\chi_R'(r^{2s}))
\\&\phantom{{}={}}
+(2s)^2(r_\theta/r)^{d-2s}\theta\chi_R''(r^{2s})r^{2s-1}
\\&= 
(d-2s)(r_\theta/r)^{d-4s}2s\theta(\chi_R'(r^{2s})-\chi_R(r^{2s})r^{-2s})r^{-1}(1+2s\theta\chi_R'(r^{2s}))
\\&\phantom{{}={}}
+(2s)^2(r_\theta/r)^{d-2s}\theta\chi_R''(r^{2s})r^{2s-1}, 
\\
\partial_r^2J
&= 
(d-2s)(d-4s)(r_\theta/r)^{d-6s}(2s\theta)^2(\chi_R'(r^{2s})-\chi_R(r^{2s})r^{-2s})^2
\\&\phantom{{}=i{}} 
\times r^{-2}(1+2s\theta\chi_R'(r^{2s}))
\\&\phantom{{}={}}
+(d-2s)(r_\theta/r)^{d-4s}(2s)^2\theta(\chi_R''(r^{2s})r^{2s}-\chi_R'(r^{2s})+\chi_R(r^{2s})r^{-2s})
\\&\phantom{{}=+i{}}
\times r^{-2}(1+2s\theta\chi_R'(r^{2s}))
\\&\phantom{{}={}}
-(d-2s)(r_\theta/r)^{d-4s}2s\theta(\chi_R'(r^{2s})-\chi_R(r^{2s})r^{-2s})r^{-2}(1+2s\theta\chi_R'(r^{2s}))
\\&\phantom{{}={}}
+(d-2s)(r_\theta/r)^{d-4s}(2s\theta)^2(\chi_R'(r^{2s})-\chi_R(r^{2s})r^{-2s})\chi_R''(r^{2s})r^{2s-2}
\\&\phantom{{}={}}
+(d-2s)(2s)^3(r_\theta/r)^{d-4s}\theta^2(\chi_R'(r^{2s})-\chi_R(r^{2s})r^{-2s})r^{-1}\chi_R''(r^{2s})r^{2s-1} 
\\&\phantom{{}={}}
+(2s)^2(r_\theta/r)^{d-2s}\theta(2s\chi_R'''(r^{2s})r^{4s-2}+(2s-1)\chi_R''(r^{2s})r^{2s-2}), 
\\ 
\tfrac{\partial r_\theta}{\partial r} 
&= 
(r_\theta/r)^{1-2s}(1+2s\theta\chi_R'(r^{2s})), 
\\ 
\tfrac{\partial^2 r_\theta}{\partial r^2} 
&= 
(1-2s)(r_\theta/r)^{1-4s}2s\theta(\chi_R'(r^{2s})-\chi_R(r^{2s})r^{-2s})r^{-1}(1+2s\theta\chi_R'(r^{2s}))
\\&\phantom{{}={}}
+(2s)^2(r_\theta/r)^{1-2s}\theta\chi_R''(r^{2s})r^{2s-1} 
\end{align*}
and 
\begin{align*}
(\tfrac{\partial r_\theta}{\partial r})(r_\theta/r)^{-1}-1
&= 
(r_\theta/r)^{-2s}(1+2s\theta\chi_R'(r^{2s}))-1 
\\&= 
\dfrac{1+2s\theta\chi_R'(r^{2s})}{1+2s\theta\chi_R(r^{2s})r^{-2s}}-1
\\&= 
\dfrac{2s\theta(\chi_R'(r^{2s})-\chi_R(r^{2s})r^{-2s})}{1+2s\theta\chi_R(r^{2s})r^{-2s}},
\\
(\tfrac{\partial r_\theta}{\partial r})^2(r_\theta/r)^{-2}-1
&= 
(r_\theta/r)^{-4s}(1+2s\theta\chi_R'(r^{2s}))^2-1 
\\&= 
\dfrac{(1+2s\theta\chi_R'(r^{2s}))^2}{(1+2s\theta\chi_R(r^{2s})r^{-2s})^2}-1
\\&= 
\dfrac{4s\theta(\chi_R'(r^{2s})+s\theta\chi_R'^2(r^{2s})-\chi_R(r^{2s})r^{-2s}-s\theta\chi_R^2(r^{2s})r^{-4s})}
{(1+2s\theta\chi_R(r^{2s})r^{-2s})^2},
\end{align*}
we can bound the first term of \eqref{eq:2210111611} as, 
by using the Cauchy--Schwarz inequality and \eqref{eq:2210121201}, 
\begin{equation}
\begin{split}\label{eq:2301181354}
&\tfrac{\hbar^2}2\mathop{\mathrm{Re}}\langle \chi^2u, 
(\tfrac{\partial r_\theta}{\partial r})^2U_\theta (-\Delta)U_\theta^{-1}u\rangle
\\&\ge 
\hbar^2(\tfrac12-C_1\beta)\|\nabla(\chi u)\|^2+O(\hbar^2)\|u\|^2,
\end{split}
\end{equation}
where $C_1>0$ is a certain constant which is independent of both $\hbar$ and $\beta$.
In the following we frequently use the Cauchy--Schwarz inequality without mentioning.

As for the second term of \eqref{eq:2210111611}, 
we evaluate it by using assumption (i) of the theorem. 
We have for any $\varepsilon\in(0,1)$ 
\begin{equation}
\begin{split}\label{eq:2302061554}
\tfrac{\hbar^2}2\mathop{\mathrm{Im}}\langle \tilde\chi^2u, \Delta u\rangle
&= 
\hbar^2\mathop{\mathrm{Im}}\langle (\nabla\chi)u, \nabla(\chi u)\rangle
\\&\ge 
-\hbar^2\left(\varepsilon\|\nabla(\chi u)\|^2 + (4\varepsilon)^{-1}\|(\nabla\chi)u\|^2\right).
\end{split}
\end{equation}
Since we have on $\mathop{\mathrm{supp}}\tilde\chi$ 
\begin{align*}
J^{-1}(\partial_rJ) 
&= 
-2s\theta(d-2s)(r_\theta/r)^{-2s}r^{-2s-1}, 
\\ 
(\tfrac{\partial r_\theta}{\partial r})^{-1}(\tfrac{\partial^2 r_\theta}{\partial r^2})
&= 
-2s\theta(1-2s)(r_\theta/r)^{-2s}r^{-2s-1},
\\
(\tfrac{\partial r_\theta}{\partial r})(r_\theta/r)^{-1}-1
&= 
(r_\theta/r)^{-2s}-1 
\end{align*}
and 
\begin{equation*}
(r_\theta/r)^{-2s}=\dfrac1{1+2s\theta r^{-2s}}=\dfrac{1-\ri 2s\beta r^{-2s}}{1+4s^2\beta^2 r^{-4s}},
\end{equation*}
it holds that 
\begin{align*}
&-\tfrac{\hbar^2}2\mathop{\mathrm{Im}}\langle \tilde\chi^2u, 
[J^{-1}(\partial_rJ)
+(\tfrac{\partial r_\theta}{\partial r})^{-1}(\tfrac{\partial^2 r_\theta}{\partial r^2})
-(d-1)\!\left((\tfrac{\partial r_\theta}{\partial r})(r_\theta/r)^{-1}-1\right)\!r^{-1}]\partial_r
u\rangle
\\&\ge 
-4s^2\hbar^2\beta^2|1-2s||\langle \tilde\chi u,  \tfrac{r^{-4s-1}}{1+4s^2\beta^2r^{-4s}}\partial_r\tilde\chi u\rangle|
+O(\hbar^2)\|u\|^2.
\end{align*}
Moreover, noting that on $\mathop{\mathrm{supp}}\tilde\chi$ 
\begin{align*}
\mathop{\mathrm{Im}}\left[(\tfrac{\partial r_\theta}{\partial r})^2(r_\theta/r)^{-2}-1\right]
&= 
\mathop{\mathrm{Im}}\dfrac{1}{(1+2s\theta r^{-2s})^2}
= 
\dfrac{-4s\beta r^{-2s}}{(1+4s^2\beta^2 r^{-4s})^2},
\end{align*}
we have 
\begin{align*}
&-\tfrac{\hbar^2}2\mathop{\mathrm{Im}}\langle \tilde\chi^2u, 
-\left((\tfrac{\partial r_\theta}{\partial r})^2(r_\theta/r)^{-2}-1\right)\!
r^{-2}\Delta_{\mathbb{S}^{d-1}}u\rangle
\\&= 
\tfrac{\hbar^2}2\langle \tilde\chi^2u, 
\tfrac{-4s\beta r^{-2s}}{(1+4s^2\beta^2 r^{-4s})^2}r^{-2}\Delta_{\mathbb{S}^{d-1}}u\rangle
\\&\ge 0.
\end{align*}
Therefore we have by letting $\varepsilon=1/4$ 
\begin{equation}\label{eq:2210121208}
\begin{split}
&-\tfrac{\hbar^2}2\mathop{\mathrm{Im}}\langle \tilde\chi^2u, 
(\tfrac{\partial r_\theta}{\partial r})^2U_\theta (-\Delta)U_\theta^{-1}u\rangle
\\&\ge 
-\tfrac14\hbar^2\|\nabla(\chi u)\|^2 
-4s^2\hbar^2\beta^2|1-2s||\langle \tilde\chi u,  \tfrac{r^{-4s-1}}{1+4s^2\beta^2r^{-4s}}\partial_r\tilde\chi u\rangle|
+O(\hbar^2)\|u\|^2.
\end{split}
\end{equation}
Here we remark that we can bound the second term on the right-hand side of \eqref{eq:2210121208} 
as, by using a bootstrap argument, 
\begin{equation}
\begin{split}\label{eq:2301201448}
&-4s^2\hbar^2\beta^2|1-2s||\langle \tilde\chi u,  \tfrac{r^{-4s-1}}{1+4s^2\beta^2r^{-4s}}\partial_r\tilde\chi u\rangle|
\\&\ge 
O(\beta^2)\mathop{\mathrm{Re}}
\langle u, \chi_R^2(r^{2s})r^{-8s-2}(\tfrac{\partial r_\theta}{\partial r})^2(H_\theta-z)u\rangle 
+O(\beta^2)\|u\|^2
\end{split}
\end{equation}
for sufficiently small $\beta>0$.
Let us confirm it. 
We have
\begin{equation}
\begin{split}\label{eq:2301201451}
&-4s^2\hbar^2\beta^2|1-2s||\langle \tilde\chi u,  \tfrac{r^{-4s-1}}{1+4s^2\beta^2r^{-4s}}\partial_r\tilde\chi u\rangle| 
\\&\ge 
-\hbar^2\beta^2\|\chi_R(r^{2s})r^{-4s-1}\partial_ru\|^2+O(\hbar^2\beta^2)\|u\|.
\end{split}
\end{equation}
By the expression of $U_\theta(-\Delta)U_\theta^{-1}$, see \eqref{eq:2301201357}, 
it holds that 
\begin{align*}
&\hbar^2\|\chi_R(r^{2s})r^{-4s-1}\partial_ru\|^2
\\&= 
\hbar^2\mathop{\mathrm{Re}}\langle \chi_R(r^{2s})r^{-4s-1}\partial_ru, \chi_R(r^{2s})r^{-4s-1}\partial_ru\rangle 
\\&\phantom{{}={}}
+\hbar^2\langle u, \chi_R^2(r^{2s})r^{-8s-2}(-r^{-2}\Delta_{\mathbb{S}^{d-1}})u\rangle 
\\&\phantom{{}={}}
-\hbar^2\langle u, \chi_R^2(r^{2s})r^{-8s-2}(-r^{-2}\Delta_{\mathbb{S}^{d-1}})u\rangle 
\\&\le 
\hbar^2\mathop{\mathrm{Re}}\langle u, \chi_R^2(r^{2s})r^{-8s-2}(-\Delta)u\rangle 
\\&\phantom{{}={}}
-\hbar^2\langle u, \chi_R^2(r^{2s})r^{-8s-2}(-r^{-2}\Delta_{\mathbb{S}^{d-1}})u\rangle 
+C_2\hbar^2\|u\|^2
\\&= 
\hbar^2\mathop{\mathrm{Re}}\langle u, \chi_R^2(r^{2s})r^{-8s-2}(\tfrac{\partial r_\theta}{\partial r})^2
[U_\theta(-\Delta)U_\theta^{-1}]u\rangle 
\\&\phantom{{}={}}
-\hbar^2\mathop{\mathrm{Re}}\langle u, \chi_R^2(r^{2s})r^{-8s-2}
[J^{-1}(\partial_rJ)+(\tfrac{\partial r_\theta}{\partial r})^{-1}(\tfrac{\partial^2 r_\theta}{\partial r^2})]\partial_ru\rangle 
\\&\phantom{{}={}}
+(d-1)\hbar^2\mathop{\mathrm{Re}}\langle u, \chi_R^2(r^{2s})r^{-8s-2}
[(\tfrac{\partial r_\theta}{\partial r})(r_\theta/r)^{-1}-1]r^{-1}\partial_ru\rangle 
\\&\phantom{{}={}}
+\hbar^2\mathop{\mathrm{Re}}\langle u, \chi_R^2(r^{2s})r^{-8s-2}
[(\tfrac{\partial r_\theta}{\partial r})^2(r_\theta/r)^{-2}-1]r^{-2}\Delta_{\mathbb{S}^{d-1}}u\rangle 
\\&\phantom{{}={}}
-\hbar^2\mathop{\mathrm{Re}}\langle u, \chi_R^2(r^{2s})r^{-8s-2}\phi(r)u\rangle 
-\hbar^2\langle u, \chi_R^2(r^{2s})r^{-8s-2}(-r^{-2}\Delta_{\mathbb{S}^{d-1}})u\rangle 
+C_2\hbar^2\|u\|^2
\\&\le 
2\mathop{\mathrm{Re}}\langle u, \chi_R^2(r^{2s})r^{-8s-2}(\tfrac{\partial r_\theta}{\partial r})^2
(H_\theta-z)u\rangle 
\\&\phantom{{}={}}
+\hbar^2|\langle u, \chi_R^2(r^{2s})r^{-8s-3}O(\beta)\partial_ru\rangle| 
\\&\phantom{{}={}}
-(1+O(\beta^2))\hbar^2\langle u, \chi_R^2(r^{2s})r^{-8s-2}(-r^{-2}\Delta_{\mathbb{S}^{d-1}})u\rangle 
+C_3\|u\|^2.
\end{align*}
Moreover we have for any $\varepsilon\in(0,1)$ 
\begin{align*}
|\langle u, \chi_R^2(r^{2s})r^{-8s-3}O(\beta)\partial_ru\rangle| 
&\le 
\varepsilon\beta^2\|\chi_R(r^{2s})r^{-4s-1}\partial_ru\|^2
+\varepsilon^{-1}C_4\|u\|^2, 
\end{align*}
and thus, for sufficiently small $\varepsilon>0$ and $\beta>0$ we have
\begin{align*}
&\hbar^2\|\chi_R(r^{2s})r^{-4s-1}\partial_ru\|^2 
\\&\le 
(1-\varepsilon\beta^2)^{-1}2\mathop{\mathrm{Re}}\langle u, \chi_R^2(r^{2s})r^{-8s-2}(\tfrac{\partial r_\theta}{\partial r})^2
(H_\theta-z)u\rangle 
\\&\phantom{{}={}}
-(1-\varepsilon\beta^2)^{-1}(1+O(\beta^2))\hbar^2
\langle u, \chi_R^2(r^{2s})r^{-8s-2}(-r^{-2}\Delta_{\mathbb{S}^{d-1}})u\rangle 
+C_5\|u\|^2
\\&\le 
(1-\varepsilon\beta^2)^{-1}2\mathop{\mathrm{Re}}\langle u, \chi_R^2(r^{2s})r^{-8s-2}(\tfrac{\partial r_\theta}{\partial r})^2
(H_\theta-z)u\rangle 
+C_5\|u\|^2.
\end{align*}
By substituting this inequality into \eqref{eq:2301201451}, we can obtain \eqref{eq:2301201448}.
By summing up \eqref{eq:2301181354}, \eqref{eq:2210121208} and \eqref{eq:2301201448}, 
one has the bound
\begin{equation}
\begin{split}\label{eq:2210121213}
&\tfrac{\hbar^2}2\mathop{\mathrm{Re}}\langle \chi^2u, 
(\tfrac{\partial r_\theta}{\partial r})^2U_\theta (-\Delta)U_\theta^{-1}u\rangle
-\tfrac{\hbar^2}2\mathop{\mathrm{Im}}\langle \tilde\chi^2u, 
(\tfrac{\partial r_\theta}{\partial r})^2U_\theta (-\Delta)U_\theta^{-1}u\rangle
\\&\ge 
\hbar^2(\tfrac14-C_1\beta)\|\nabla(\chi u)\|^2 
+O(\beta^2)\mathop{\mathrm{Re}}
\langle u, \chi_R^2(r^{2s})r^{-8s-2}(\tfrac{\partial r_\theta}{\partial r})^2(H_\theta-z)u\rangle 
\\&\phantom{{}={}}
+O(\hbar^2)\|u\|^2
+O(\beta^2)\|u\|^2
\\&\ge 
O(\beta^2)\mathop{\mathrm{Re}}
\langle u, \chi_R^2(r^{2s})r^{-8s-2}(\tfrac{\partial r_\theta}{\partial r})^2(H_\theta-z)u\rangle 
\\&\phantom{{}={}}
+O(\hbar^2)\|u\|^2
+O(\beta^2)\|u\|^2
\end{split}
\end{equation}
for $\beta>0$ small enough. 

Next we bound the third term of \eqref{eq:2210111611}. 
Using the Taylor expansion of $q_\theta$ in $\theta$ 
\begin{equation*}
q_\theta=q+r^{1-2s}\chi_R(r^{2s})(\partial_rq)\theta+O(\theta^2) 
\end{equation*}
and assumption (ii) of the theorem, we see that 
\begin{equation}
\begin{split}\label{eq:2210121343}
&\mathop{\mathrm{Re}}\langle \chi u, 
(\tfrac{\partial r_\theta}{\partial r})^2(-\tfrac12r_\theta^{2s}+q_\theta-z)\chi u\rangle
\\&= 
\langle \chi u, (-\tfrac12r^{2s}+q-E +O(\beta^2))\chi u\rangle
\\&\ge 
\alpha\|\chi u\|^2+O(\beta^2)\|u\|^2.
\end{split}
\end{equation}
Similarly, by using assumption (iii) of the theorem we have 
\begin{equation}
\begin{split}\label{eq:2210121345}
&-\mathop{\mathrm{Im}}\langle \tilde\chi u, 
(\tfrac{\partial r_\theta}{\partial r})^2(-\tfrac12r_\theta^{2s}+q_\theta-z)\tilde\chi u\rangle
\\&= 
-\langle \tilde\chi u, 
[-(1-2s)+2(1-2s)r^{-2s}(q-E) -(s-r^{1-2s}(\partial_rq)-\mu)]\beta\tilde\chi u\rangle
\\&\phantom{{}={}}
+O(\beta^2)\|u\|^2 
\\&\ge 
\beta\gamma\|\tilde\chi u\|^2 
+O(\beta^2)\|u\|^2. 
\end{split}
\end{equation}
Hence by the expression \eqref{eq:2210111611} 
and the bounds \eqref{eq:2210121213}, \eqref{eq:2210121343} and \eqref{eq:2210121345}, it holds that 
\begin{align*}
\mathop{\mathrm{Re}}\langle v, (H_\theta-z)u\rangle
&\ge 
(\min\{\alpha, \beta\gamma\}+O(\hbar^2)+O(\beta^2))\|u\|^2
\\&\phantom{{}={}}
+O(\beta^2)\mathop{\mathrm{Re}}
\langle u, \chi_R^2(r^{2s})r^{-8s-2}(\tfrac{\partial r_\theta}{\partial r})^2(H_\theta-z)u\rangle. 
\end{align*}
We are done.
\end{proof}

\subsection{Inverted harmonic oscillator}\label{inverted}

This is the last section and we give a proof of Theorem~\ref{thm:no-resonance} for $s=1$. 
Since the proof is quite similar to that for $s\in(0,1)$, 
we sometimes omit computations. 

\begin{proof}[Proof of Theorem~\ref{thm:no-resonance} with $s=1$]
Let $\theta=\ri\beta$. 
Take any $u\in\vD(H_\theta)$, $\|u\|=1$. 
Then we set $v=\overline{(\tfrac{\partial r_\theta}{\partial r})}^2(\chi^2-\ri\tilde\chi^2)u$. 
Note that $\|v\|\le (1+O(\beta))\|u\|$.
We show that for $z=E-\ri\beta\mu$
\begin{equation}
\begin{split}\label{eq:lower-bound-2}
&\mathop{\mathrm{Re}}\langle v, (H_\theta-z)u\rangle 
\\&\ge 
(\min\{\alpha, \beta\gamma\}+O(\hbar^2)+O(\beta^2))\|u\|^2
+\beta^2\mathop{\mathrm{Re}}\langle u, \psi(r)(H_\theta-z)u\rangle, 
\end{split}
\end{equation}
where $\psi(r)$ is a certain bounded function. 
Then by combining with the inequalities
\begin{align*}
(1+O(\beta))\|u\|\|(H_\theta-z)u\| 
&\ge 
\|v\|\|(H_\theta-z)u\| 
\ge 
\mathop{\mathrm{Re}}\langle v, (H_\theta-z)u\rangle, 
\\ 
O(\beta^2)\|u\|\|(H_\theta-z)u\| 
&\ge 
-\beta^2\mathop{\mathrm{Re}}\langle u, \psi(r)(H_\theta-z)u\rangle, 
\end{align*}
we have the assertion.

The left-hand side of \eqref{eq:lower-bound-2} is expressed as 
\begin{equation}
\begin{split}\label{eq:2210111611-2}
\mathop{\mathrm{Re}}\langle v, (H_\theta-z)u\rangle
&=
\tfrac{\hbar^2}2\mathop{\mathrm{Re}}\langle \chi^2u, 
(\tfrac{\partial r_\theta}{\partial r})^2U_\theta (-\Delta)U_\theta^{-1}u\rangle
\\&\phantom{{}={}}
-\tfrac{\hbar^2}2\mathop{\mathrm{Im}}\langle \tilde\chi^2u, 
(\tfrac{\partial r_\theta}{\partial r})^2U_\theta (-\Delta)U_\theta^{-1}u\rangle
\\&\phantom{{}={}}
+\mathop{\mathrm{Re}}\langle \chi u, 
(\tfrac{\partial r_\theta}{\partial r})^2(-\tfrac12r_\theta^2+q_\theta-z)\chi u\rangle
\\&\phantom{{}={}}
-\mathop{\mathrm{Im}}\langle \tilde\chi u, 
(\tfrac{\partial r_\theta}{\partial r})^2(-\tfrac12r_\theta^2+q_\theta-z)\tilde\chi u\rangle.
\end{split}
\end{equation}
Note that the expression \eqref{eq:2301201357} also holds even for the case of $s=1$. 
Since we have \eqref{eq:2210121201}, 
\begin{align*}
\partial_rJ 
  &= 
(d-2)(r_\theta/r)^{d-4}\theta(-2r^{-3}\log r\chi_R(r^2)+r^{-3}\chi_R(r^2)+2r^{-1}\log r\chi_R'(r^2)) 
\\&\phantom{{}=i{}}
\times (1+r^{-2}\chi_R(r^2)\theta+2\log r\chi_R'(r^2)\theta)
\\&\phantom{{}={}}
+(r_\theta/r)^{d-2}\theta(-2r^{-3}\chi_R(r^2)+4r^{-1}\chi_R'(r^2)+4r\log r\chi_R''(r^2)),
\\
\partial_r^2J
  &= 
(d-2)(d-4)(r_\theta/r)^{d-6}\theta^2(2r^{-3}\log r\chi_R(r^2)-r^{-3}\chi_R(r^2)-2r^{-1}\log r\chi_R'(r^2))^2 
\\&\phantom{{}=i{}}
\times (1+r^{-2}\chi_R(r^2)\theta+2\log r\chi_R'(r^2)\theta)
\\&\phantom{{}={}}
+(d-2)(r_\theta/r)^{d-4}\theta\left[(6\log r-5)r^{-4}\chi_R(r^2)-(6\log r-4)r^{-2}\chi_R'(r^2)\right. 
\\&\phantom{{}=+i{}}
\left. +4\log r\chi_R''(r^2)\right]\!\left(1+r^{-2}\chi_R(r^2)\theta+2\log r\chi_R'(r^2)\theta\right)
\\&\phantom{{}={}}
+(d-2)(r_\theta/r)^{d-4}\theta^2(2r^{-3}\log r\chi_R(r^2)-r^{-3}\chi_R(r^2)-2r^{-1}\log r\chi_R'(r^2))
\\&\phantom{{}=+i{}}
\times (2r^{-3}\chi_R(r^2)-4r^{-1}\chi_R'(r^2)-4r\log r\chi_R''(r^2))
\\&\phantom{{}={}}
+(d-2)(r_\theta/r)^{d-4}\theta^2
(2r^{-3}\log r\chi_R(r^2)-r^{-3}\chi_R(r^2)-2r^{-1}\log r\chi_R'(r^2))
\\&\phantom{{}=+i{}}
\times (2r^{-3}\chi_R(r^2)-4r^{-1}\chi_R'(r^2)-4r\log r\chi_R''(r^2))
\\&\phantom{{}={}}
+(r_\theta/r)^{d-2}\theta\left[6r^{-4}\chi_R(r^2)-8r^{-2}\chi_R'(r^2)+12\chi_R''(r^2)+4\log r\chi_R''(r^2)\right.  
\\&\phantom{{}=+i{}}
\left. +8r^2\log r\chi_R'''(r^2)\right]\!,
\\ 
\tfrac{\partial r_\theta}{\partial r} 
  &= 
(r_\theta/r)^{-1}(1+r^{-2}\chi_R(r^2)\theta+2\log r\chi_R'(r^2)\theta), 
\\ 
\tfrac{\partial^2 r_\theta}{\partial r^2} 
  &= 
-(r_\theta/r)^{-3}\theta
(-2r^{-3}\log r\chi_R(r^2)+r^{-3}\chi_R(r^2)+2r^{-1}\log r\chi_R'(r^2)) 
\\&\phantom{{}=+i{}}
\times (1+r^{-2}\chi_R(r^2)\theta+2\log r\chi_R'(r^2)\theta) 
\\&\phantom{{}={}}
-(r_\theta/r)^{-1}\theta(2r^{-3}\chi_R(r^2)-4r^{-1}\chi_R'(r^2)-4r\log r\chi_R''(r^2))
\end{align*}
and 
\begin{align*}
&(\tfrac{\partial r_\theta}{\partial r})(r_\theta/r)^{-1}-1
  = 
\dfrac{\theta(r^{-2}\chi_R(r^2)+2\log r\chi_R'(r^2)-2r^{-2}\log r\chi_R(r^2))}
{1+2r^{-2}\log r\chi_R(r^2)\theta},
\\
&(\tfrac{\partial r_\theta}{\partial r})^2(r_\theta/r)^{-2}-1
  \\&= 
\theta\left[ r^{-4}\chi_R^2(r^2)\theta+4(\log r)^2\chi_R'^2(r^2)\theta+2r^{-2}\chi_R(r^2)
+4\log r\chi_R'(r^2)\right.
\\&\phantom{{}=+i{}}
+4r^{-2}\log r\chi_R(r^2)\chi_R'(r^2)\theta-4r^{-2}\log r\chi_R(r^2) 
\\&\phantom{{}=+i{}}
-\left. 4r^{-4}(\log r)^2\chi_R(r^2)^2\theta\right]
/(1+2r^{-2}\log r\chi_R^2(r^2)\theta)^2,
\end{align*}
we can bound the first term of \eqref{eq:2210111611-2} as 
\begin{equation}
\begin{split}\label{eq:2301181354-2}
&\tfrac{\hbar^2}2\mathop{\mathrm{Re}}\langle \chi^2u, 
(\tfrac{\partial r_\theta}{\partial r})^2U_\theta (-\Delta)U_\theta^{-1}u\rangle
\\&\ge 
\hbar^2(\tfrac12-C_1\beta)\|\nabla(\chi u)\|^2+O(\hbar^2)\|u\|^2,
\end{split}
\end{equation}
where $C_1>0$ is a certain constant which is independent of both $\hbar$ and $\beta$.

Next we bound the second term of \eqref{eq:2210111611-2}. 
Noting that we have on $\mathop{\mathrm{supp}}\tilde\chi$ 
\begin{align*}
J^{-1}(\partial_rJ) 
  &= 
-\theta(d-2)(r_\theta/r)^{-2}r^{-3}(2\log r-1), 
\\ 
(\tfrac{\partial r_\theta}{\partial r})^{-1}(\tfrac{\partial^2 r_\theta}{\partial r^2})
  &= 
\theta(r_\theta/r)^{-2}r^{-3}(2\log r-1)-2r^{-3}\theta/(1+r^{-2}\theta),
\\
(\tfrac{\partial r_\theta}{\partial r})(r_\theta/r)^{-1}-1
  &= 
(r_\theta/r)^{-2}(1+r^{-2}\theta)-1, 
\\
(\tfrac{\partial r_\theta}{\partial r})^2(r_\theta/r)^{-2}-1
  &= 
(r_\theta/r)^{-4}(1+r^{-2}\theta)^2-1 
\end{align*}
and 
\begin{equation*}
(r_\theta/r)^{-2}=\dfrac1{1+2\theta r^{-2}\log r}
=\dfrac{1-\ri 2\beta r^{-2}\log r}{1+4\beta^2 r^{-4}(\log r)^2},
\end{equation*}
we have 
\begin{align*}
&-\tfrac{\hbar^2}2\mathop{\mathrm{Im}}\langle \tilde\chi^2u, 
[J^{-1}(\partial_rJ)
+(\tfrac{\partial r_\theta}{\partial r})^{-1}(\tfrac{\partial^2 r_\theta}{\partial r^2})
-(d-1)\!\left((\tfrac{\partial r_\theta}{\partial r})(r_\theta/r)^{-1}-1\right)\!r^{-1}]\partial_r
u\rangle
\\&\ge 
-2\hbar^2\beta^2|\langle \tilde\chi u, 
\tfrac{r^{-5}\log r(2\log r-1)}{1+4\beta^2r^{-4}(\log r)^2}\partial_r\tilde\chi u\rangle|
-\hbar^2\beta^2|\langle \tilde\chi u, 
\tfrac{r^{-5}}{1+\beta^2r^{-4}}\partial_r\tilde\chi u\rangle|
+O(\hbar^2)\|u\|^2
\end{align*}
and, by using assumption (i), 
\begin{align*}
&-\tfrac{\hbar^2}2\mathop{\mathrm{Im}}\langle \tilde\chi^2u, 
-\left((\tfrac{\partial r_\theta}{\partial r})^2(r_\theta/r)^{-2}-1\right)\!
r^{-2}\Delta_{\mathbb{S}^{d-1}}u\rangle
\\&\ge 
-C_1\hbar^2\beta^3\langle \tilde\chi^2u, r^{-6}(\log r)(-r^{-2}\Delta_{\mathbb{S}^{d-1}})u\rangle.
\end{align*}
By a bootstrap argument we can obtain the bounds, for sufficiently small $\beta>0$, 
\begin{align*}
&-2\hbar^2\beta^2|\langle \tilde\chi u, 
\tfrac{r^{-5}\log r(2\log r-1)}{1+4\beta^2r^{-4}(\log r)^2}\partial_r\tilde\chi u\rangle|
\\&\ge 
-\beta^2\mathop{\mathrm{Re}}\langle u, \psi_1(r)(H_\theta-z)u\rangle 
+O(\beta^2)\|u\|^2, 
\\
&-\hbar^2\beta^2|\langle \tilde\chi u, 
\tfrac{r^{-5}}{1+\beta^2r^{-4}}\partial_r\tilde\chi u\rangle|
\\&\ge 
-\beta^2\mathop{\mathrm{Re}}\langle u, \psi_2(r)(H_\theta-z)u\rangle 
+O(\beta^2)\|u\|^2 
\end{align*}
and 
\begin{align*}
&-C_1\hbar^2\beta^3\langle \tilde\chi^2u, r^{-6}(\log r)(-r^{-2}\Delta_{\mathbb{S}^{d-1}})u\rangle
\\&\ge 
-\beta^3\mathop{\mathrm{Re}}\langle u, \psi_3(r)(H_\theta-z)u\rangle 
+O(\beta^2)\|u\|^2. 
\end{align*}
Here ${\psi_j}'s$ are uniformly bounded smooth functions, although they depend on $\beta>0$. 
Therefore by \eqref{eq:2302061554} with $\varepsilon=1/4$, we have 
\begin{equation}
\begin{split}\label{eq:2302061743}
&-\tfrac{\hbar^2}2\mathop{\mathrm{Im}}\langle \tilde\chi^2u, 
(\tfrac{\partial r_\theta}{\partial r})^2U_\theta (-\Delta)U_\theta^{-1}u\rangle
\\&\ge 
-\beta^2\mathop{\mathrm{Re}}\langle u, \psi_4(r)(H_\theta-z)u\rangle 
-\tfrac{\hbar^2}4\|\nabla(\chi u)\|^2
+O(\hbar^2)\|u\|^2 
+O(\beta^2)\|u\|^2, 
\end{split}
\end{equation}
where $\psi_4=\psi_1+\psi_2+\beta\psi_3$.
Let us bound the third and the fourth term of \eqref{eq:2210111611-2} 
by using the assumptions (ii) and (iii) of the theorem, respectively.
Using the Taylor expansion of $q_\theta$ in $\theta$ 
\begin{equation*}
q_\theta=q+r^{-1}\log r\chi_R(r^2)(\partial_rq)\theta+O(\theta^2), 
\end{equation*}
we have 
\begin{equation}
\begin{split}\label{eq:2210121343-2}
\mathop{\mathrm{Re}}\langle \chi u, 
(\tfrac{\partial r_\theta}{\partial r})^2(-\tfrac12r_\theta^2+q_\theta-z)\chi u\rangle
&= 
\langle \chi u, (-\tfrac12r^2+q-E +O(\beta^2))\chi u\rangle
\\&\ge 
\alpha\|\chi u\|^2+O(\beta^2)\|u\|^2.
\end{split}
\end{equation}
Similarly, we can see that 
\begin{equation}
\begin{split}\label{eq:2210121345-2}
&-\mathop{\mathrm{Im}}\langle \tilde\chi u, 
(\tfrac{\partial r_\theta}{\partial r})^2(-\tfrac12r_\theta^2+q_\theta-z)\tilde\chi u\rangle
\\&= 
-\langle \tilde\chi u, 
\left[-2r^{-2}\log r(-\tfrac12r^2+q-E)+2r^{-2}(-\tfrac12r^2+q-E) \right. 
\\&\phantom{{}=+++{}}
\left. +(-\log r+r^{-1}\log r (\partial_rq)+\mu)\right]\!\beta\tilde\chi u\rangle
+O(\beta^2)\|u\|^2 
\\&\ge 
\beta\gamma\|\tilde\chi u\|^2 
+O(\beta^2)\|u\|^2. 
\end{split}
\end{equation}
Hence by the expression \eqref{eq:2210111611-2} 
and the bounds 
\eqref{eq:2301181354-2}--\eqref{eq:2210121343-2} and \eqref{eq:2210121345-2}, we finally have 
\begin{align*}
\mathop{\mathrm{Re}}\langle v, (H_\theta-z)u\rangle
&\ge 
(\min\{\alpha, \beta\gamma\}+O(\hbar^2)+O(\beta^2))\|u\|^2
\\&\phantom{{}={}}
-\beta^2\mathop{\mathrm{Re}}
\langle u, \psi_4(r)(H_\theta-z)u\rangle 
\end{align*}
for sufficiently small $\beta>0$, and then the assertion \eqref{eq:lower-bound-2} follows.

\end{proof}

\section*{Acknowledgments}
The author would like to thank K.~Higuchi and K.~Ito for helpful discussions.
The author is supported by JSPS KAKENHI Grant Number 19H05599.


\end{document}